\documentclass[10pt,conference,a4paper,twocolumn]{IEEEtran}
\usepackage{graphicx,subfigure}
\usepackage{amsmath,amssymb,amsthm}
\usepackage{enumerate}
\usepackage{algorithm,algorithmic}
\usepackage{mdwlist}
\usepackage{cite} % [1]-[5], instead of [1],[2],[3],[4],[5]
\usepackage{flushend,cuted}

\newcommand{\figref}[1]{Fig.\,\ref{#1}}

\newtheorem{Theorem}{\textbf{Theorem}}
\newtheorem{Corollary}{\textbf{Corollary}}
\newtheorem{Definition}{\textbf{Definition}}

{\proof}{\proofend}

\def\a{\alpha}
\def\mA{\boldsymbol A}
\def\b{\beta}
\def\c{\mathbf c}

\def\Dmin{D_{\min}}

\def\F{\mathbb F}

\def\P{\mathcal{P}}
\def\p{\mathbf{p}}
\def\R{\mathcal{R}}
\def\r{r}

\def\Umin{U_{\min}}
\def\W{\mathcal W}

\author{
	\IEEEauthorblockN{Mingchao Yu and Parastoo Sadeghi}
	\IEEEauthorblockA{Research School of Engineering, The Australian National University, Canberra, Australia\\
		Emails: \{ming.yu,~parastoo.sadeghi\}@anu.edu.au}}
\title{Approximating Throughput and Packet Decoding Delay in Linear Network Coded Wireless Broadcast}
\begin{document}
\maketitle
\begin{abstract}
In this paper, we study a wireless packet broadcast system that uses linear network coding (LNC) to help receivers recover data packets that are missing due to packet erasures. We study two intertwined performance metrics, namely throughput and average packet decoding delay (APDD) and establish strong/weak approximation relations based on whether the approximation holds for the performance of every receiver (strong) or for the average performance across all receivers (weak). We prove an equivalence between strong throughput approximation and strong APDD approximation. We prove that throughput-optimal LNC techniques can strongly approximate APDD, and partition-based LNC techniques may weakly approximate throughput. We also prove that memoryless LNC techniques, including instantly decodable network coding techniques, are not strong throughput and APDD approximation nor weak throughput approximation techniques.
%
%
%
%
%When using linear network coding (LNC) techniques to help receivers recover missing data packets in a packet-block based lossy wireless broadcast system, the two most concerned performance metrics are how fast the coded broadcast can be complete (throughput) and how fast each data packet can be decoded by each receiver (average packet decoding delay, APDD). Although both metrics have been extensively optimized separately, the interplay between their optimization is not well understood. In this paper, we introduce the concept of strong/weak approximation (which is more general than optimization) of throughput and APDD based on whether the approximation holds for the performance of every receiver or for the average performance across all receivers. We prove an equivalence between strong throughput approximation and strong APDD approximation. We prove that throughput optimal LNC techniques can strongly approximate APDD, and partition-based LNC techniques may weakly approximate throughput. We also prove that memoryless LNC techniques, including instantly decodable network coding techniques, are not strong throughput and APDD approximation nor weak throughput approximation techniques.
\end{abstract}
\section{Introduction}
In this paper, we consider a wireless broadcast problem where a sender wishes to broadcast a block $\P$ of $K$ data packets to a set of $N$ receivers using linear network coding (LNC) \cite{li2003linear,koetter2003algebraic}. Each receiver is assumed to already possess a subset of $\P$ and still wants all the remaining data packets.

For such systems, two important performance metrics are throughput and average packet decoding delay (APDD). While throughput measures how fast the broadcast can be completed, APDD measures how fast each individual data packet can be decoded by each receiver. A lower APDD implies faster data delivery to the application layer on average, and is particularly important when individual data packets are informative.

Throughput can be maximized if every LNC coded packet is innovative to every receiver who has not fully recovered $\P$. Such packets can be generated either randomly (i.e., the classic random LNC (RLNC) technique \cite{ho:medard:koetter:karger:effros:2006}) or deterministically (e.g., by solving a hitting set problem \cite{kwan2011generation}, or by adding extra data packets to instantly decodable coded packets \cite{keller2008online}). However, the APDD performance of these techniques have not been well studied. Recently, Yu \emph{et al} proved that RLNC approximates the minimum APDD average over all receivers with a ratio of 2 \cite{yu:sprintson:sadeghi:netcod2015}, i.e., its APDD is at most two times of the minimum.

\begin{figure}[t]
	\centering
	\includegraphics[width=0.8\linewidth]{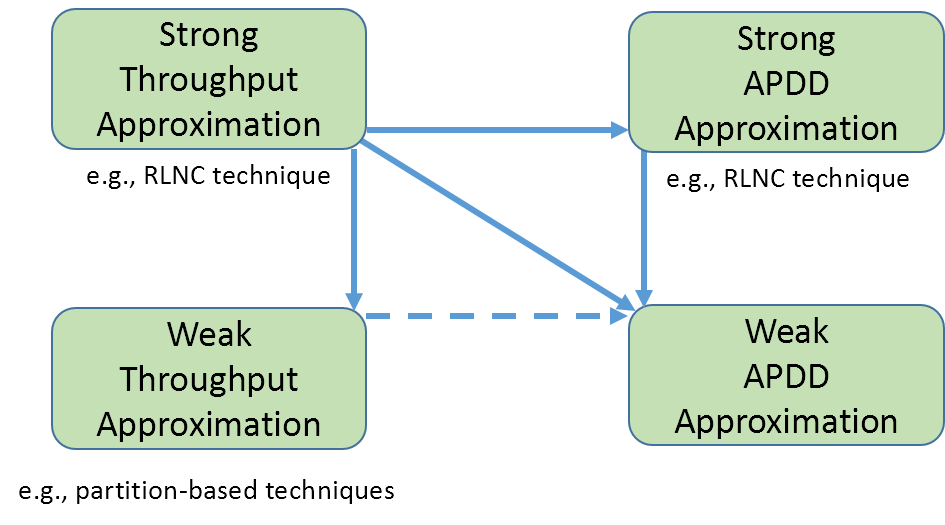}
	\caption{Interplay between throughput and APDD approximation.}
	\label{fig:interplay}
\end{figure}

APDD can be minimized if every LNC coded packet allows every receiver to instantly decode a wanted data packet \cite{yu:sprintson:sadeghi:netcod2015}. Such coded packets, however, are NP-hard to find \cite{yu:sprintson:sadeghi:netcod2015}. Instead, instantly decodable network coding (IDNC) techniques generate in each transmission a coded packet that allows a subset of receivers to instantly decode a wanted packet, and ask the remaining receivers to discard this coded packet rather than storing it in the memory for future decoding. Thus, IDNC techniques are memoryless. Due to this feature, IDNC techniques are generally not throughput optimal \cite{li:idnc_video:2011}.  It has also been proved in \cite{sorour2015completion} that it is intractable to maximize the throughput of general IDNC techniques. Although a large body of heuristics have been developed as a remedy, it is an open problem whether memoryless LNC techniques are able to approximate the optimal throughput and APDD.

A tradeoff between throughput and APDD can be achieved by partitioning $\P$  into disjoint sub-blocks and broadcasting them separately using certain LNC techniques \cite{maymounkov:generation:2006,emina:li:isit2012,yu2013rapprochement,yu:parastoo:neda:2014}. Although it is understood that such partitioned-based LNC techniques are generally not throughput and APDD optimal, their approximation performance has not been studied.

Due to the fact that throughput and APDD optimization could be intractable and that heuristics cannot guarantee bounded performance, throughput and APDD approximation is important in the design and evaluation of LNC techniques. However, to the best of our knowledge, LNC throughput and APDD approximation has not received much attention in the literature. Moreover, the interplay between throughput and APDD optimization has not been well studied. We also note that optimization is a special approximation with a ratio of 1.

Therefore, in this paper, we study the more general problem of throughput and APDD approximation in linear network coded wireless broadcast. Specifically, we will
\begin{enumerate}
\item introduce the concepts of strong/weak throughput and APDD approximation. Here strong (resp. weak) means that the approximation holds for the performance of every receiver (resp. averaged over all receivers);
\item investigate the interplay between throughput and APDD approximation;
\item evaluate the approximation performance of the aforementioned three classes of LNC techniques.
\end{enumerate}
Some main findings of this paper are (also depicted in \figref{fig:interplay}):
\begin{itemize}
\item A strong throughput $\beta$-approximation technique also strongly approximates APDD with a ratio of at most $2\beta$. This relation does not necessarily hold between weak throughput and APDD approximation;
\item All strong throughput-optimal LNC techniques strongly approximate APDD with a ratio between $\frac{4}{3}$ and $2$;
\item A technique that partitions the packet block into $M$ disjoint sub-blocks and applies a weak throughput $\beta$-approximation technique to each sub-block weakly approximates throughput with a ratio of at most $\beta M$;
\item Memoryless LNC techniques are not strong throughput and APDD approximation nor weak throughput approximation techniques.
\end{itemize}

\section{System Model and Performance Measurements}
%\subsection{System Model}
We consider a block-based wireless broadcast scenario, in which the sender wishes to deliver a block of $K$ data packets, denoted by $\P=\{\p_k\}_{k=1}^K$, to a set of $N$ receivers, denoted by $\R=\{\r_n\}_{n=1}^N$. All data packets are vectors of the same length, with entries taken from a finite field $\F_q$. Time is slotted, and in each time slot the sender broadcasts a coded packet to all receivers. The wireless channel between the sender and each receiver $\r_n$ is independent of each other, and is subject to Bernoulli random packet erasures with a probability of $P_{e,n}$.

We assume each receiver has already received a subset of packets in $\P$ and still wants all the rest. Such a packet reception state could be the consequence of previous uncoded transmissions \cite{heide_systematic_RLNC}, and is a common assumption in network coding and index coding literature \cite{birk2006coding}.
This state can be summarized by a binary $N\times K$ state feedback matrix (SFM) $\mA$: $\mA(n,k)=1$ means $\r_n$ has missed $\p_k$ and wants it, and $\mA(n,k)=0$ means $\r_n$ already has $\p_k$.
The set of data packets wanted by $r_n$ is
denoted by $\W_n$. Its size is denoted by $w_n$.
%An example of $\mA$ is given in \figref{fig:sfm}. There are 6 data packets and 3 receivers. Every receiver wants 3 data packets.

The sender then applies an LNC technique to help receivers recover their missing data packets. In each time slot, it broadcasts an LNC packet $\c$, which takes the form of:
\begin{equation}
	\c=\sum_{\p_k\in\P}\a_k\p_k,
\end{equation}
where $\{\a_k\}$ are coding coefficients chosen from $\F_q$. In particular, RLNC technique chooses $\{\a_k\}$ uniformly at random.

%%\captionsetup{font={small,stretch=1.5
%\begin{figure}
%	\centering
%\includegraphics[width=0.5\linewidth]{./figures/SFM_uniform}
%	\caption{An example of state feedback matrix $\mA$. \textbf{This can go}}
%	\label{fig:sfm}
%\end{figure}

\def\mysum{\mathsf{sum}}
\subsection{Performance Metrics}
Our first performance metric is throughput. It measures how fast the broadcast of $\P$ can be finished. Noting that a minimum of $K$ unocded transmissions is initially needed, we measure throughput by the total number of coded transmissions in the broadcast, which is denoted by $U$. Clearly, a smaller $U$ indicates higher throughput. We further denote by $U_n$ the number of coded transmissions after which receiver $\r_n$ decodes all its wanted data packets. Consequently, $U=\max(\{U_n\}_{n=1}^N)$.

Our second performance metric is average packet decoding delay (APDD). It reflects how fast each individual data packet is decoded by each receiver on average. Given a realization of $\mA$, the APDD of receiver $\r_n$, denoted by $D_n$, is:
\begin{equation}\label{eq:d_def:rec}
	D_n=\frac{1}{w_n}\sum_{\forall k: \p_k\in\W_n}u_{n,k},
\end{equation}
where $u_{n,k}$ is the index of the coded transmission after which $\r_n$ decodes $\p_k$. The APDD across all receivers is similar:
\begin{equation}\label{eq:d_def}
	D=\frac{1}{\mysum(\mA)}\sum_{\forall k,n: \mA(n,k)=1}u_{n,k},
	\vspace{-0.5em}
\end{equation}
where $\mysum(\mA)$ is the sum of the entries of $\mA$, and is equal to the number of ones in $\mA$.

\subsection{Performance Limits and Expectations}
We denote by $U_{\min,n}$ (resp. $\Umin$) the minimum possible $U_n$ (resp. $U$) that any LNC techniques can offer without packet erasures. It is clear that $U_{\min,n}=w_n$ and $\Umin=\max(\{w_n\}_{n=1}^N)$. We further denote by $\overline U_{\min,n}$  (resp. $\overline U_{\min}$) the minimum expected $U_n$ (resp. $U$) that any LNC techniques can offer with random packet erasures. It is clear that $\overline U_{\min,n}\geqslant U_{\min,n}$ and $\overline U_{\min}\geqslant \Umin$, and the equalities hold when there are no packet erasures. $\overline U_{\min}$ has been studied in the literature through studying RLNC \cite{nistor2011delay}.

Similarly, we denote by $D_{\min,n}$ (resp. $D_{\min}$) the minimum $D_n$ (resp. $D$) that any LNC techniques can offer without packet erasures. It holds that $D_{\min,n}\geqslant \frac{w_n+1}{2}$ and $D_{\min}\geqslant \frac{\sum w_nD_{\min,n}}{\sum w_n}$, where the equality holds when every coded packet allows every receiver to instantly decode a wanted data packet. We further denote by $\overline D_{\min,n}$  (resp. $\overline D_{\min}$) the minimum expected $D_n$ (resp. $D$) that any LNC techniques can offer with random packet erasures. Again, $\overline D_{\min,n}\geqslant D_{\min,n}$ and $\overline D_{\min}\geqslant \Dmin$. It is proved in \cite{yu2015minimizing} that $\overline D_{\min}$ is NP-hard to find.

If an LNC technique called ``X'' is applied, we add $(X)$ to the end of the above. For example, $\overline D_{\min,n}(\mathrm{RLNC})$ denotes the minimum expected APDD $D_n$ of $\r_n$ under RLNC.

\section{Defining Performance Approximation}
In this section, we define strong and weak approximation of throughput and APDD.

\subsection{Strong Approximation}
We define strong throughput approximation as follows:

\begin{Definition}
An LNC technique X is a strong throughput $\b$-approximation technique if and only if:
\begin{equation}
\frac{\overline U_{\min,n}(\mathrm{X})}{\overline U_{\min,n}}\leqslant \beta
\end{equation}
for every receiver $\r_n$ under any SFM and any packet erasure probabilities $\{P_{e,n}\}_{n=1}^N$, where $\beta\geqslant 1$ is a constant. In particular, if $\beta=1$, then technique-X is a strong throughput-optimal technique. 
\end{Definition}

According to this definition, when a strong throughput $\b$-approximation LNC technique is applied, every receiver $\r_n$ can expect to decode all its wanted data packets within $\beta \overline U_{\min,n}$ coded transmissions regardless of the packet reception state and packet erasure probability of the other receivers.

We define strong APDD approximation similarly:

\begin{Definition}
An LNC technique X is a strong APDD $\b$-approximation technique if and only if:
\begin{equation}
		\frac{\overline D_{\min,n}(\mathrm{X})}{\overline D_{\min,n}}\leqslant \beta
\end{equation}
for every receiver $\r_n$ under any SFM and any packet erasure probabilities $\{P_{e,n}\}_{n=1}^N$, where $\beta\geqslant 1$ is a constant. In particular, if $\beta=1$, then technique-X is a strong APDD-optimal technique. 
\end{Definition}

\subsection{Weak Approximation}
We define weak throughput approximation as follows:

\begin{Definition}
	An LNC technique X is a weak throughput $\b$-approximation technique if and only if:
	\begin{equation}
		\frac{\overline U_{\min}(\mathrm{X})}{\overline U_{\min}}\leqslant \beta
	\end{equation}
	for any SFM and any packet erasure probabilities $\{P_{e,n}\}_{n=1}^N$, where $\beta\geqslant 1$ is a constant. In particular, if $\beta=1$, then technique-X is a weak throughput -ptimal technique. 
\end{Definition}

According to this definition, with a weak throughput $\b$-approximation LNC technique, we can expect to complete the coded broadcast within $\beta \overline U_{\min}$ coded transmissions.

Similarly, we define weak APDD approximation as follows:

\begin{Definition}
An LNC technique X is a weak APDD $\b$-approximation technique if and only if:
\begin{equation}
		\frac{\overline D_{\min}(\mathrm{X})}{\overline D_{\min}}\leqslant \beta
\end{equation}
for any SFM and any packet erasure probabilities $\{P_{e,n}\}_{n=1}^N$, where $\beta\geqslant 1$ is a constant. In particular, if $\beta=1$, then technique-X is a weak APDD-optimal technique. 
\end{Definition}

It is clear that a strong throughput/APDD approximation technique is also a weak one, but not necessarily vice versa. Our main interest in this paper is the interplay between throughput  and APDD approximation. To this end, we first establish the performance of a reference technique, namely RLNC, that strongly approximates both throughput and APDD.

\section{The Performance of RLNC}
In this section, we study the approximation performance of RLNC, and then extend the result to the general class of throughput-optimal LNC techniques.

\begin{Theorem}
RLNC is a strong throughput-optimal and strong APDD 2-approximation technique. Mathematically, for every receiver $\r_n$, it always holds that $\overline U_{\min,n}(\mathrm{RLNC})=\overline U_{\min,n}$ and $\overline D_{\min,n}(\mathrm{RLNC})\leqslant2\overline D_{\min,n}$.
\end{Theorem}

\begin{proof}
It is clear that
\begin{equation}
\overline U_{\min,n}(\mathrm{RLNC})=\overline U_{\min,n}=\frac{w_n}{1-P_{e,n}},
\end{equation}
Then, since $\r_n$ decodes all its wanted data packets on average after $\overline U_{\min,n}(\mathrm{RLNC})$ coded transmissions, we have
\begin{equation}
\overline D_{\min,n}(\mathrm{RLNC})=\overline U_{\min,n}(\mathrm{RLNC})=\frac{w_n}{1-P_{e,n}}.
\end{equation}
On the other hand, the authors of \cite{yu2015minimizing} has proved that
\begin{equation}
\overline D_{\min,n}\geqslant \frac{w_n+1}{2(1-P_{e,n})},
\end{equation}
Therefore,
\vspace{-0.5em}
\begin{equation}
\frac{\overline D_{\min,n}(\mathrm{RLNC})}{\overline D_{\min,n}} \leqslant\frac{2w_n}{w_n+1}\leqslant 2,
\end{equation}
which completes the proof.
\end{proof}
We note that in terms of APDD performance, RLNC is in fact, the worst technique in the class of strong throughput-optimal LNC techniques, as it generally does not provide early packet decodings (excluding occasional early decodings). Thus, we can state that all LNC techniques in this class strongly approximate APDD with a ratio of at most 2. This result can be further strengthened into the following:
\begin{Theorem}
All strong throughput-optimal LNC techniques strongly approximate APDD with a ratio between $\frac{4}{3}$ and 2.
\end{Theorem}
\begin{proof}
%\textbf{Not clear}To prove the lower bound $\frac{4}{3}$, we  an instance of the SFM for which $\frac{4}{3}$ is the minimum approximation ratio any strong throughput-optimal LNC technique, say technique-X, has $D_{\min,n}(\mathrm X)=\frac{4}{3}D_{\min,n}$ for at least one receiver $\r_n$ \textbf{?} when there are no packet erasures.
%Here we only need to show an instance of SFM for which an APDD approximation ratio of $4/3$ is achieved. This will indicates that the approximation ratio
Since the approximation ratio of an LNC technique is the largest ratio it provides across any SFMs, to prove that the ratio is at least $\frac{4}{3}$ for strong throughput-optimal techniques we only need an instance of SFM where $\frac{4}{3}$ is achieved by them.

Our SFM consists of $2$ data packets and 3 receivers. $\r_1$ only wants $\p_1$, $\r_2$ only wants $\p_2$, and $\r_3$ wants both packets. For this SFM, any strong throughput-optimal LNC technique-X will send as $\c_1$ a linear combination of $\p_1$ and $\p_2$ to satisfy both $\r_1$ and $\r_2$. However, $\c_1$ does not allow $\r_3$ to decode. $\r_3$ can only decode after the second coded transmission. Thus, $D_{\min,3}(X)=2$. On the other hand, by sending $\p_1$ and $\p_2$ separately, $D_{\min,3}=1.5$. Thus, $\frac{U_{\min,3}(X)}{U_{\min,3}}=\frac{4}{3}$.
\end{proof}

\section{Interplay Between Throughput and APDD Approximation}
With the help of RLNC, we establish the following relation between throughput and APDD approximation:
\begin{Theorem}
	Strong throughput $\b$-approximation techniques strongly approximate APDD with a ratio of at most $2\b$.
\end{Theorem}
\begin{proof}
	Consider a strong throughput $\b$-approximation LNC technique called X. By definition,
	\begin{equation}
		\overline U_{\min,n}(\mathrm{X})\leqslant \b \overline U_{\min,n},
	\end{equation}
	for any receiver $\r_n$ in any given SFM. Then, since
	\begin{equation}\nonumber
	\begin{cases}
	\overline D_{\min,n}(\mathrm{X})\leqslant \overline U_{\min,n}(\mathrm{X}), &~\\
	\overline U_{\min,n} = \overline U_{\min,n}(\mathrm{RLNC}) = \overline D_{\min,n}(\mathrm{RLNC}),~\mathrm{and}\\
	\overline D_{\min,n}(\mathrm{RLNC}) \leqslant \overline 2D_{\min,n},
	\end{cases}
	\end{equation}
	we obtain $\overline D_{\min,n}(\mathrm{X})\leqslant 2\beta D_{\min,n}$.
\end{proof}

On the other hand, weak throughput approximation techniques do not necessarily weakly approximate APDD. To see this, we will prove in the next section that partition-based LNC techniques may weakly approximate throughput but may not weakly approximate APDD. We summarize the interplay between throughput and APDD approximation in \figref{fig:interplay}.

\section{Partition-based LNC Techniques}
Given an SFM, partition-based LNC techniques partition the packet block $\P$ into $M$ ($M>1$) disjoint or overlapped sub-blocks $\{\P_m\}_{m=1}^M$ \cite{maymounkov:generation:2006,silva:generation:2009,emina:li:isit2012,yu2013rapprochement,yu:parastoo:neda:2014}. In this paper, we only consider the disjoint case. Correspondingly, the SFM is partitioned into $M$ sub-SFMs $\{\mA_m\}_{m=1}^M$, where in $\mA_m$, the $N$ receivers want data packets from  $\P_m$. A certain LNC technique (e.g. RLNC) is then applied to each sub-block separately in order.

%In this section, we will show that the approximation performance of partition-based LNC techniques heavily depends on the value of $M$.
%
%
%\begin{Theorem}
%A partition-based LNC technique that partitions the packet block into an unbounded number $M$ of sub-blocks is not strong throughout and APDD approximation technique.
%\end{Theorem}
%\begin{proof}
%"Unbounded" implies that, for any given $\beta$, there exists an SFM such that this partition-based LNC technique will partition it into $M>\beta$ sub-SFMs. 
%\end{proof}

\begin{Theorem}
%Partition-based LNC technique are not strong throughput and APDD approximation techniques. However,
If a partition-based LNC technique applies a weak throughput $\beta$-approximation LNC technique to each of the $M$ sub-SFMs, then it is at most a weak throughput $2\beta M$-approximation technique.
\end{Theorem}

\begin{proof}
When a weak throughput $\beta$-approximation technique called X is applied to any given SFM $\mA$, by definition it holds that $\overline U_{\min}(X)\leqslant \beta \overline U_{\min}$. Since any sub-SFM of $\mA$ requires, on average, at most $\overline U_{\min}(X)$ coded transmissions, any $M$-partition of $\mA$ need at most $M\overline U_{\min}(X)$ coded transmissions. Thus, $\overline U_{\min}(\mathrm{partition,~X})\leqslant \beta M\overline U_{\min}$.
\end{proof}
However, weakly approximating throughput may not help these techniques weakly approximate APDD:
\begin{Theorem}\label{theo:partition_weak_throughput}
Partition-based throughput weak approximation techniques do not necessarily weakly approximate APDD.
\end{Theorem} 

\begin{proof}
For any given SFM $\mA$, without loss of generality let us assume receiver $\r_1$ wants the largest subset $\W_1$ of data packets of $\P$. Consider a partition-based technique, called X, that partitions $\P$ into two sub-blocks: $\P_1=\W_1$ and $\P_2=\P\setminus\W_1$. Due to Theorem \ref{theo:partition_weak_throughput}, when RLNC is applied, technique-X can weakly approximate the throughput with a ratio of $2$.

We now show that technique-X cannot weakly approximate APDD. Consider an SFM with $N$ receivers. $\r_1$ wants $\W_1$, and all the remaining $N-1$ receivers only want one data packet not in $\W_1$. When $N\gg w_1$, we have $D_{\min}\approx 1$. But if technique-X is applied, the remaining $N-1$ receivers can only decode after $\r_1$ has fully decoded, indicating that $D_{\min}(\mathrm{X})\approx w_1+1$, which is not within a constant multiple of $D_{\min}$.\footnote{
For this particular SFM, there exist better partition strategies that are able to minimize APDD. This, however, is irrelevant to the theorem and its proof.}
\end{proof}

This theorem also indicates the general relation between weak throughput and APDD approximation:
\begin{Corollary}
Weak throughput approximation techniques are not necessarily weak APDD approximation techniques.
\end{Corollary}

However, we are not able to identify the strong throughput and APDD approximation performance of partition-based LNC techniques without specifying the partitioning strategy, which is out of the scope of this paper.
\section{Memoryless LNC Techniques}
An LNC technique is memoryless if its receivers discard undecodable coded packet(s) rather than storing them for future decodings. A well-known class of memoryless LNC techniques is IDNC, which allows a subset of receivers to instantly decode a wanted data packet from each coded packet, so that APDD could be reduced.
However, the cost is a degradation in the throughput of receivers who discard useful, but instantly undecodable coded packets. In this section, we prove the following two theorems:

%It is yet unclear whether such performance limits are able to approximate the global optimal throughput and APDD of LNC techniques has not been studied.

\begin{Theorem}\label{theo:mem_weak_thpt}
Memoryless LNC techniques are not weak throughput approximation techniques.
\end{Theorem}
We prove this theorem in the appendix by showing that, for an SFM where every pair of two data packets is wanted by a different receiver, memoryless LNC techniques require at least $\lceil \log_2 K\rceil+1$ coded transmissions, which is not within a constant multiple of $U_{\min}=2$. (Here $\lceil x\rceil$ is the smallest integer greater than or equal to $x$.)

Then, since every strong throughput approximation technique is a weak one, the above theorem indicates that:

\begin{Corollary}
Memoryless LNC techniques are not strong throughput approximation techniques.
\end{Corollary}

The proof of Theorem \ref{theo:mem_weak_thpt} also sheds some light on the APDD approximation performance of memoryless LNC techniques.
\begin{Theorem}
Memoryless LNC techniques are not strong APDD approximation techniques.
\end{Theorem}
\begin{proof}
In the proof of Theorem \ref{theo:mem_weak_thpt}, receivers who decode their second wanted data packet after the last coded transmission have an APDD of at least $\left\lceil\log_2K\right\rceil/2+1$. However, $D_{\min,n}\leqslant D_{\min,n}(\mathrm{RLNC})=2$. Thus, $D_{\min,n}(\mathrm{memoryless})$ is not within a constant multiple of $D_{\min,n}$.
\end{proof}
We summarize our results on the approximation performance of the three classes of LNC techniques in Table \ref{tab:performance}.

\section{Conclusion}
In this paper, we generalized the problem of throughput and APDD optimization in linear network coded wireless broadcast to their strong and weak approximations. This generalization fills the gap between optimal and heuristic LNC techniques with approximation techniques, such as RLNC (strong throughput optimal and strong APDD 2-approximation) and partition-based LNC techniques (weak throughput approximation). By using these LNC techniques as references, we also revealed the interplay between throughput and APDD approximation, including a relation between strong throughput $\beta$-approximation and strong APDD $2\beta$-approximation, as well as the independence between weak throughput and APDD approximation. Besides, we negated the strong and weak throughput approximation and the strong APDD approximation of memoryless LNC techniques. Our results could inspire new approaches to design and evaluate LNC techniques.

As future work, we wish to tackle the interplay between strong APDD approximation and weak throughput approximation, the strong throughput and APDD approximation performance of partition-based LNC techniques, and the weak APDD approximation performance of memoryless LNC techniques. We are also interested in extending our research to index coding \cite{birk2006coding}, as well as and other applications of LNC, such as cooperative data exchange \cite{sprintson2010randomized}.

%Besides, there are some interesting open problems we wish to tackle in the future, such as the interplay between strong APDD approximation and weak throughput approximation, the strong throughput and APDD approximation performance of partition-based LNC techniques, and the weak APDD approximation performance of memoryless LNC techniques. We are also interested in extending our research to index coding \cite{sprintson:algorithm:2008,sprintson:ic_nc_matroid:2010,Yossef:index:2011} and other applications of LNC, such as cooperative data exchange \cite{el2010coding,sprintson2010randomized}.

\newcommand{\tabincell}[2]{\begin{tabular}{@{}#1@{}}#2\end{tabular}}
\begin{table}
\caption{The approximation performance of three classes of LNC.}
\label{tab:performance}
\begin{tabular}{|c|c|c|c|c|}
\hline
~& \tabincell{c}{Strong \\ Throughput} & \tabincell{c}{Strong\\ APDD} & \tabincell{c}{Weak \\ Throughput} & \tabincell{c}{Weak \\ APDD}\\
\hline
RLNC&yes&yes&yes&yes\\
\hline
Partition-based & open & open & yes & may not be\\\hline 
Memoryless LNC&no&no&no&open\\
\hline
\end{tabular}
\end{table}

\begin{appendices}

\section{Proof of Theorem 6}
Our proof involves two types of SFMs:
\begin{itemize}
\item $\mA_1(K)$: every pair of data packets is wanted by a different receiver. There are $N=\frac{K(K-1)}{2}$ receivers;
\item $\mA_2(K,m)$: every data packet is wanted by $m$ different receivers. Every pair of data packets is wanted by a different receiver. There are $N=mK+\frac{K(K-1)}{2}$ receivers.
\end{itemize}

Note that $\Umin=2$ for $\mA_1(K)$. We prove the theorem by proving that $\Umin(\mathrm{memoryless})=\left\lceil\log_2K\right\rceil+1$ for $\mA_1(K)$.

%Transmission starts by XOR-ing an arbitrary $m_1 \ge 1$ data packets in $\mA_1(K)$.  As theThese $m_1$ packets will not be instantly decodable to receivers that want a (different) pair of these data packets, but will be instantly decodable for $(K-m_1)m_1$ receivers that only want one of these data packets. Hence, the resulting SFM is split into two sub-SFMs: $\mA_1(m_1)$ and $\mA_2(K-m_1, m_1)$. These two sub-SFMs are independent in the sense that a coded packet of $\mA_1(m_1)$ and a coded packet of $\mA_2(K-m_1,m_1)$ can be XOR-ed and sent without affecting their decodability for receivers.

The transmission starts by sending as $\c_1$ the XOR of any $m_1 \ge 1$ data packets in $\mA_1(K)$. The resulted SFM consists of two sub-SFMs: 1) an $A_1(K-m_1)$, which contains the $m_1$ data packets and the receivers who want 2 data packets from $\c_1$ and thus discard $\c_1$; 2) an $A_2(K-m_1, m-1)$, which contains the remaining $K-m_1$ data packets and the remaining receivers, which either has decoded one wanted data packet from $\c_1$ and still want one data packet from $K-m_1$, or want 2 data packets from $K-m_1$.
%the receivers who want one data packet from $\c_1$ an will discard this packet; 2) receivers who want 1 data packets from $\c_1$ will instantly decode them, and still want one of the remaining $K-m_1$ data packets; 3) receivers who want 2 data packets from the remaining $K-m$ data packets will discard this packet. Thus, the resulted SFM consists of two sub-SFMs: the $m_1$ data packets and the receivers in 1) form $A_1(m_1)$, while the remaining $K-m_1$ data packets and the receivers in 2) and 3) form $A_2(K-m_1, m_1)$.
These two sub-SFMs are independent in the sense that a coded packet of $\mA_1(m_1)$ and a coded packet of $\mA_2(K-m_1,m_1)$ can be XOR-ed and sent without affecting their decodability for receivers.

Similarly, we can show that after sending the XOR of any arbitrary $m_3$ data packets from $A_2(K-m_1, m_1)$, the resulted SFM consists of two independent sub-SFMs: an $A_1(m_3)$ and an $A_2(K-m_1-m_3,m_1+m_3)$.

Continuing the logic, after the $u$-th transmission ($u\geqslant 1$), $\mA_1(K)$ is split into $2^{u-1}$ type-1 sub-SFMs and $2^{u-1}$ type-2 sub-SFMs. Only sub-SFMs that consists of a single data packet can be completed in one coded transmission and be removed. The evolution of $\mA_1(K)$ is demonstrated in a layered graph in figure \ref{fig:memoryless_decoding}. The $u$-th layer corresponds to the SFM before the $u$-th coded transmission. The total number of coded transmissions is thus the number of layers plus one. It is clear that the minimum number of layers is $\left\lceil\log_2K\right\rceil$, which is achieved by XOR-ing half of the data packets from each sub-SFM. Thus, $\Umin(\mathrm{memoryless})=\left\lceil\log_2K\right\rceil+1$.

%Upon the reception of a linear combination (denoted by $\oplus$) of any $m$ data packets from $\mA_1(K)$ and after memoryless decoding  by all receivers, the resulted SFM consists of two independent sub-SFMs: an $\mA_1(m)$ and an $\mA_2(K-m, m)$. Here by independent, we mean receivers in the two sub-SFMs do not want data packets from each other. Consequently, a coded packet of $\mA_1(m)$ and a coded packet of $\mA_2(K-m,m)$ can be XOR-ed and sent without affecting their decodability.
%
%Upon the reception of the $\oplus$ of any $m_1$ data packets from $\mA_2(K,m)$ and after memoryless decoding by all receivers, the resulted SFM also consists of two independent sub-SFMS: an $\mA_1(m_1)$ and $\mA_2(K-m_1,m+m_1)$.
%
%Therefore, after the $u$-th transmission ($u\geqslant 1$), $\mA_1(K)$ is split into $2^{u-1}$ type-1 sub-SFMs and $2^{u-1}$ type-2 sub-SFMs. Only sub-SFMs that consists of a single data packet can be completed in one coded transmission and be removed. The evolution of $\mA_1(K)$ is demonstrated in a layered graph in figure \ref{fig:memoryless_decoding}. The $u$-th layer corresponds to the SFM before the $u$-th coded transmission. The total number of coded transmissions is thus the number of layers plus one.
%
%It is clear that the minimum number of layers is $\left\lceil\log_2K\right\rceil$, which is achieved by XOR-ing half of the data packets from each sub-SFM. Thus, $\Umin(\mathrm{memoryless})=\left\lceil\log_2K\right\rceil+1$.

\begin{figure}[t]
	\includegraphics[width=\linewidth]{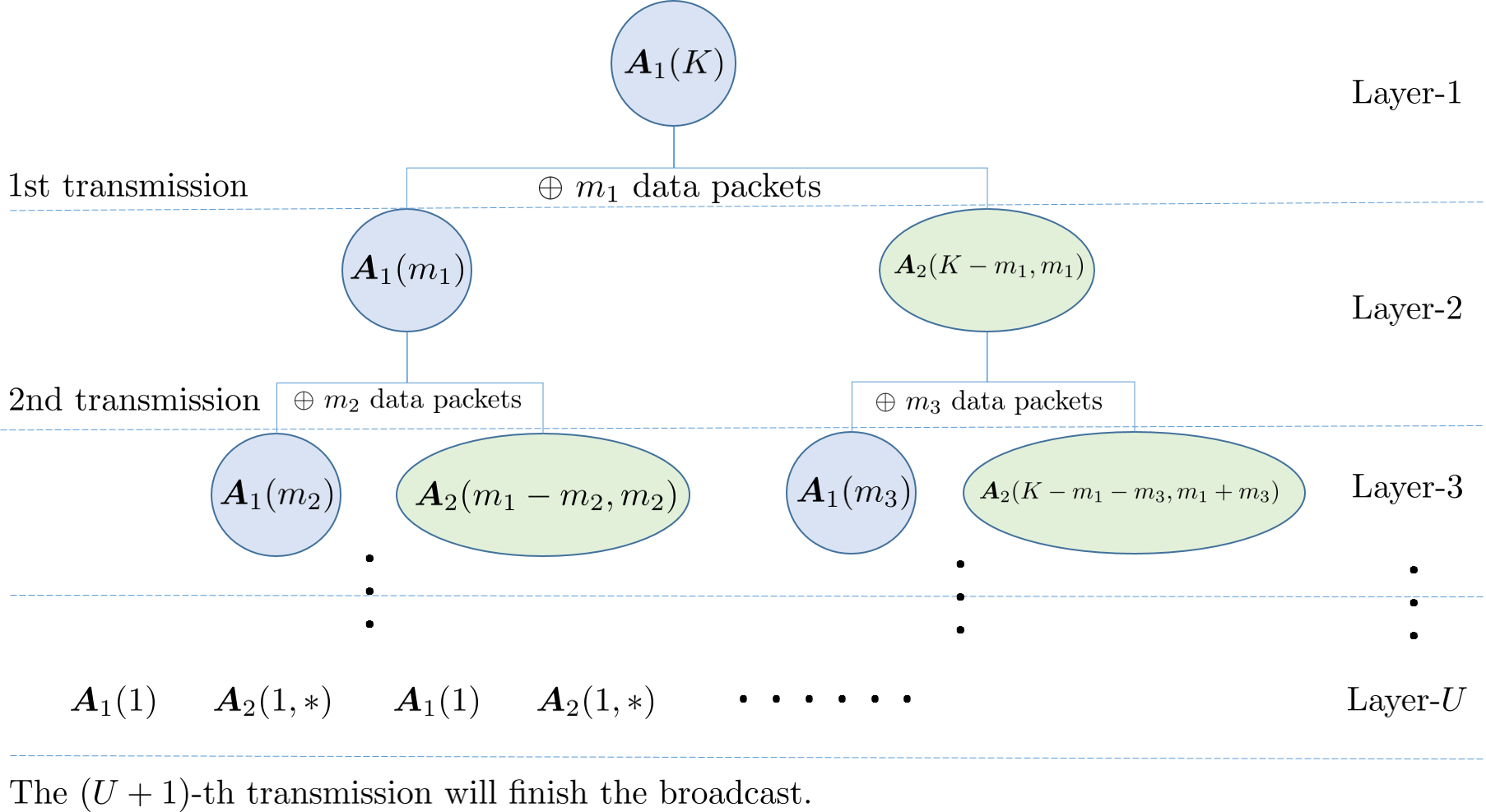}
	\caption{The decoding evolution of $\mA_1(K)$. (Note that '*'s are arbitrary positive integers whose values depend on the coded packets.)}
	\label{fig:memoryless_decoding}
\end{figure}
\end{appendices}

\bibliographystyle{IEEEtran}
\bibliography{IEEEabrv,My_ref}
%\footnotesize\section*{Acknowledgment}This work was supported under Australian Research Council Discovery Projects funding scheme (project no. DP120100160).

\end{document}